\documentclass[12pt]{amsart}

\usepackage[numbers,square,comma]{natbib}
\usepackage{float}
\usepackage[pdfpagelabels]{hyperref}
\usepackage[nointegrals]{wasysym}
\usepackage{amsmath}
\usepackage{amssymb}
\usepackage{amsthm}
\usepackage{amsaddr}
\usepackage{graphicx}
\usepackage{geometry}
\usepackage{tikz}
\usepackage{tkz-graph}
\usepackage{chngcntr}
\usepackage{thmtools}
\usepackage[linesnumbered,lined,boxed,commentsnumbered]{algorithm2e}
\usepackage{thm-restate}
\usepackage{tabularx}
\usepackage{environ}
\usepackage{algorithm2e}
\usetikzlibrary{shapes}

\counterwithin{figure}{section}

\makeatletter
\newcommand{\problemtitle}[1]{\gdef\@problemtitle{#1}}% Store problem title
\newcommand{\probleminput}[1]{\gdef\@probleminput{#1}}% Store problem input
\newcommand{\problemquestion}[1]{\gdef\@problemquestion{#1}}% Store problem question
\NewEnviron{problem}{
  \problemtitle{}\probleminput{}\problemquestion{}% Default input is empty
  \BODY% Parse input
  \par\addvspace{.5\baselineskip}
  \noindent
  \begin{tabularx}{\textwidth}{@{\hspace{\parindent}} l X c}
    \multicolumn{2}{@{\hspace{\parindent}}l}{\@problemtitle} \\% Title
    \textbf{Input:} & \@probleminput \\% Input
    \textbf{Question:} & \@problemquestion% Question
  \end{tabularx}
  \par\addvspace{.5\baselineskip}
}
\makeatother

\theoremstyle{plain}
\newtheorem{theorem}{Theorem}[section]
\newtheorem{definition}[theorem]{Definition}

\newtheorem{claim}[theorem]{Claim}
\newtheorem{example}[theorem]{Example}

\newtheorem{proposition}[theorem]{Proposition}

\newtheorem*{openproblem}{Problem}

\newcommand{\CSP}{\operatorname{CSP}}

\title[Complexity of Maltsev CSPs]{Maltsev Constraint  Satisfaction Problems and Deterministic Logspace With Counting}

\author[D.Delic]{Dejan Delic}
\address{Department of Mathematics, Toronto Metroplitan University}
\email{ddelic@torontomu.ca}

\author[A.Syed]{Ali Syed}
\address{Department of Mathematics, Toronto Metroplitan University}
\email{ali.mortoza.syed@torontomu.ca}

\begin{document}

\begin{abstract}

In this article, we prove that the problem of solving $\operatorname{CSP}(\mathbf{A})$, where $\mathbf{A}$ is a finite relational template which admits a Maltsev polymorphism is in a specific complexity class \textsc{DET}, which is related to the complexity of computing the determinant of a matrix with integer entries. Such a class is intimately related to well-studied \textsc{MOD}-logspace classes in the theory of computational complexity. To prove this fact, we develop a new algorithm for solving syntactically simple binary instances of Maltsev constraint satisfaction problems, rather different from the well-known Bulatov-Dalmau algorithm, which does not require the explicit use or knowledge of a Maltsev polymorphism of the template but, rather, utilizes a graph whose vertices are  2-generated subuniverses of $\mathbb{A}$, where $\mathbb{A}$ is the Maltsev algebra parametrizing $\operatorname{CSP}(\mathbf{A})$. The theoretical importance of this algorithm is reflected in two facts: (1) it places the problem $ \operatorname{CSP}(\mathbf{A})$ in a complexity class related to the deterministic logspace with counting, which, in itself,  has a strong connection to a variety of standard algorithmic problems in linear algebra, and (2) it only makes use of the relational structure of the template without the need for the explicit use of a compatible Maltsev polymorphism, depending entirely on the strong ``symmetry" of constraints compatible with such polymorphisms and the knowledge of 2-generated subuniverses of $\mathbb{A}$.
\end{abstract}


\maketitle

\section{Introduction}
One of the fundamental problems in constraint programming and, more widely, in the field of artificial intelligence, is the problem of determining the computational complexity of constraint satisfaction problems (CSPs, for short). The problem, in its full generality, is \textsc{NP}-complete but many of its subclasses are tractable with algorithms which are well understood. In this paper, we adopt the following convention to studying the complexity of CSPs: we study the restrictions of  the instances by allowing a fixed set of constraint relations. This approach is generally referred to in the literature as a \emph{constraint language} or, a \emph{fixed template} \cite{b-j-k}. This point of view has lead to significant progress in the study of the complexity of constraint satisfaction in the past 30 years or so. The culmination of this approach was realized in recent proofs of the Dichotomy Conjecture in \cite{bulatov2017} and \cite{zhuk2017}.

One of the early landmark results in this direction were the proof of A. Bulatov   and a substantially simplified version of the same result by  A. Bulatov and V. Dalmau (\cite{bulatov2006simple}) concerning the existence of such an algorithm for a fairly general, yet rather natural class of parametrized CSPs, those whose parametrizing algebra is \emph{Maltsev}, i.e. in which the constraint language is invariant under an algebraic operation satisfying the condition $m(x,x,y)\approx m(y,x,x)\approx y$, for all elements $x$ and $y$ of the algebra. Their algorithm, also known as the Generalized Gaussian Elimination, provided a common generalization for already known algorithms for solving CSPs over affine domains, CSPs on finite groups with near subgroups, etc. The result of Bulatov and Dalmau is based on the algebraic fact that, given any Maltsev algebra $\mathbb{A}$, any subpower of $\mathbb{A}^n$ has a generating set of polynomial (in fact, linear) size in $n$. This approach was further generalized in \cite{BIMMVW} to show that a modification of Bulatov-Dalmau algorithm solves all CSPs over the so called domains \emph{with few subpowers}, i.e. for all the domains $\mathbb{A}$ with the property that any subpower of $\mathbb{A}^n$ has a generating set of polynomial size in $n$. 

The algorithms presented in \cite{bulatov2006simple} and \cite{BIMMVW} require explicit knowledge of the algebraic operations witnessing the few subpowers property. This, in itself, may be viewed as problematic, since, in practical applications, the CSP is generally presented in the form of its constraint language and computing the required term is a highly nontrivial problem in terms of its complexity. Secondly, those algorithms do not provide ``short'' proofs of unsatisfiability in the same way local consistency checks do. 

The algorithm we present in this article is largely driven by the attempts of the authors to address some of these issues and gain better understanding of how the structural theory of Maltsev varieties influences the solvability of Maltsev CSPs. The algorithm was also motivated by the desire to refine the existing results on the computational complexity of solving Maltsev CSPs and to provide evidence that such problems can be solved via iterated use of algorithms for solving systems of linear equations over rings of the form $\mathbb{Z}_k$, for $k\geq 2$, which are uniquely determined by the parametrizing Maltsev algebra.

From the point of view of complexity theory, the only complexity result pertaining to the Bulatov-Dalmau algorithm is that it operates in polynomial time. The algorithm is based on the fact that all constraints viewed in a natural way  as subalgebras of $\mathbb{A}^n$ are finitely generated and the algorithm orders the constraints $\{C_1,\ldots,C_l\}$, and it computes  generating sets of  $C_1$, $C_1\cap C_2$, etc, until a generating set (which may be empty) is computed for $C_1\cap\ldots\cap C_l$. The algorithm we present in this article is based on the logspace reduction of the instance to a binary instance on which a sequence of logspace reductions is performed, using logspace transducers which also have access to a specified number (depending only on $\mathbb{A}$) of \textsc{MOD}-logspace oracles. The algorithm presented here, as well as the proof of its membership in the complexity class \textsc{DET}, provides strong evidence that constraint satisfaction problems with Maltsev templates are expressible in the extensions of the first-order logic which extend the fixed-point logic and also have the ability to solve systems of linear equations over finite rings $\mathbb{Z}_p$, ($p$ - prime),  in which all equations involve at most two variables.

In addition, our algorithm only depends on the fact that the given relational template is known to be compatible with a Maltsev operation, without any need for the explicit use of such an operation besides the knowledge of 2-generated subalgebras and their congruences. For instance, the Graph Isomorphism Problem restricted to the bounded colour class instances is known to be logspace reducible to binary instances of group (and, therefore, Maltsev) CSPs, parametrized by groups of all permutations $S_l$, $l\geq 2$, on a finite set.  For such CSPs, algorithms which are based on the consistency checks in the relational template (based on e.g. reachability in an undirected graph) are more amenable,  from the computational complexity point of view, than those which make heavy use of the algebraic structure of the parametrizing algebra. 

\subsection{Organization of the paper}

In Section \ref{s: prelim}, we present the basic concepts and definitions related to CSPs, as well as some of the basic concepts from universal algebra which will be used extensively in the paper. In Section \ref{complexity}, we introduce the background notions from  computational complexity, which lead us to formulate  the main result  of the article, Theorem 3.2. In particular, we define the complexity classes $\textsc{MOD}_k\textsc{L}$ ($k\geq 1$) as well as the class \textsc{DET} of problems which can be $\textsc{NC}^1$-reduced to the problem of computing the determinant of an integer matrix. Section \ref{datalog} introduces the deterministic logspace consistency check, which will play an important role in the proof of the main result, namely, the notion of the canonical symmetric (1,2)-Datalog program. Finally, Section \ref{outline} contains a proof of the Theorem 3.2., consisting of the construction of an algorithm for solving syntactically simple binary instances of a Maltsev constraint satisfaction problem, which is in the complexity class \textsc{DET}.

\section{Preliminaries}\label{s: prelim}

\subsection{Constraint Satisfaction Problem}\label{CSP}
The notion of a constraint satisfaction problem provides us with a natural framework for a variety of problems which require simultaneous satisfiability of a number of conditions on a given set of variables. More formally,

\begin{definition} An \emph{instance} of the CSP is a triple $\mathcal{I}=(V,A,\mathcal{C})$, where $V=\{x_1,\ldots,x_n\}$ is a finite set of \emph{variables}, $A$ is a finite domain for the variables in $V$, and $\mathcal{C}$ is a finite set of \emph{constraints} of the form $C=(S,R_S)$, where $S$, the \emph{scope} of the constraint, is a $k$-tuple of variables $(x_{i_1},\ldots,x_{i_k})\in V^k$ and $R_S$ is a $k$-ary relation $R_S\subseteq A^k$, called the \emph{constraint relation} of $C$.

A \emph{solution} for the instance $\mathcal{I}$ is any assignment $f:V\rightarrow A$ such that, for every constraint $C=(S,R_S)$ in $\mathcal{C}$, $f(S)\in R_S$.
\end{definition}

From this point on, we will usually assume that the set of variables $V$ for an instance is an initial segment of the set of positive integers, i.e. $V=\{1,2,\ldots,n\},$
for some $n\geq 1$.

A relational structure $\mathbf{A}=(A, \Gamma)$, defined over the domain $A$ of the instance $\mathcal{I}$, where $\Gamma$ is a finite set of relations on $A$, is often referred to as a \emph{constraint language}, and the relations from $\Gamma$ form the signature of $\mathbf{A}$. An instance of $\CSP (\mathbf{A})$ will be an instance of the CSP such that all constraint relations belong to $\mathbf{A}$.

\subsection{Basic Algebraic Concepts}\label{sec:Algebra}
\noindent In this subsection, we introduce concepts from universal algebra which will be used in the remainder of the paper. Two good references for a more in-depth overview of universal algebra are \cite{Burris1981} and \cite{bergman}.

An \emph{algebra} is an ordered pair $\mathbb{A}=(A, F)$, where $A$ is a nonempty set, the \emph{universe} of $\mathbb{A}$, while $F$ is the set of \emph{basic operations} of $\mathbb{A}$, consisting of functions of arbitrary, but finite, arities on $A$. The list of function symbols and their arities is the \emph{signature} of $\mathbb{A}$. 

A \emph{subuniverse} of the algebra $\mathbb{A}$ is a nonempty subset $B\subseteq A$ closed under all operations of $\mathbb{A}$. If $B$ is a subuniverse of $\mathbb{A}$, by restricting all operations of $\mathbb{A}$ to $B$, such a subuniverse is a \emph{subalgebra} of $\mathbb{A}$, which we denote $\mathbb{B}\leq \mathbb{A}$.

In this article, we will be particularly interested in those subuniverses of the algebra, which have no non-trivial proper subuniverses themselves. Such subuniverses are called \emph{minimal} subuniverses, or \emph{minimal algebras.}

If $\mathbb{A}_i$ is an indexed family of algebras of the same signature, the product $\prod_i \mathbb{A}_i$ of the family is the algebra whose universe is the Cartesian products of their universes $\prod_i A_i$ endowed with the basic operations which are coordinate-wise products of the corresponding operations in $\mathbb{A}_i$. If $\mathbb{A}$ is an algebra, its $n$-th Cartesian power will be denoted $\mathbb{A}^n$.

An equivalence relation $\alpha$ on the universe $A$ of an algebra $\mathbb{A}$ is a \emph{congruence} of $\mathbb{A}$, if $\alpha \leq \mathbb{A}^2$, i.e. if $\alpha$ is preserved by all operations in the signature of $\mathbb{A}$.That is, if $f: A^k\rightarrow A$ is a $k$-ary ($k\geq 1$) operation in the signature of $\mathbb{A}$, then, if $(a_1,b_1),\ldots,(a_k,b_k)\in \alpha$, then $(f(a_1,\ldots,a_k),f(b_1,\ldots,b_k))\in \alpha$.

 In that case, one can define the algebra $\mathbb{A}/\alpha$, the \emph{quotient of} $\mathbb{A}$ \emph{by} $ \alpha$, with the universe consisting of all equivalence classes (cosets) in $A/\alpha$ and whose basic operations are induced by the basic operations of $\mathbb{A}$. The $\alpha$-congruence class containing $a\in A$ will be denoted $a/\alpha$. For an algebra $\mathbb{A}$, we denote the set of all congruences of $\mathbb{A}$, $\operatorname{Con}(\mathbb{A})$.

An algebra $\mathbb{A}$ is said to be \emph{simple} if its only congruences are the trivial, diagonal relation $0_\mathbb{A}=\{(a,a)\, \vert \, a\in A\}$ and the full relation $1_\mathbb{A}=\{ (a,b)\, \vert\, a,b\in A\}$. 

Minimal algebras, which are also simple algebras, are said to be \emph{strictly simple.} In the case of algebras all of whose operations are idempotent, it is easily seen that every congruence class is a subuniverse. Therefore, in the context of idempotent algebras, every minimal algebra is also simple, since it cannot have any congruences besides the trivial one and the full relation.

Any subalgebra of a Cartesian product of algebras $\mathbb{A}\leq \prod_i \mathbb{A}_{i\in I}$ is equipped with a family of congruences arising from projections on the product coordinates. We denote $\pi_i$ the congruence obtained by identifying the tuples in $A$ which have the same value in the $i$-th coordinate. Given any $J\subseteq I$, we can define a subalgebra of $\mathbb{A}$, $proj_J(\mathbb{A})$, which consists of the projections of all tuples in $A$ to the coordinates from $J$. If $\mathbb{A}\leq \prod_{i\in I} \mathbb{A}_i$ is such that $proj_i (\mathbb{A})=\mathbb{A}_i$, for every $i\in I$, we say that $\mathbb{A}$ is a \emph{subdirect product} and denote this fact $\mathbb{A}\leq_{sp} \prod_{i\in I} \mathbb{A}_i$.

Given an algebra $\mathbb{A}$, a \emph{term} is a syntactical object describing a composition of basic operations of $\mathbb{A}$. A \emph{term operation} $t^\mathbb{A}$ of $\mathbb{A}$ is the interpretation of the syntactical term $t(x_1,\ldots,x_m)$ as an $m$-ary operation on $A$, according to the formation tree of $t$.

 Let
$f$ be an $n$-ary operation on $A$ and let $k>0$. We write $f^{(k)}$
to denote the $n$-ary operation on $A^k$, obtained by applying $f$ coordinate-wise. That is, we define the $n$-ary operation $f^{(k)}$ on $A^k$ by
\[
f^{(k)}(\mathbf a^1,\dots,\mathbf
a^n)=(f(a^1_1,\dots,a^n_1),\dots,f(a^1_k,\dots,a^n_k)),
\]
for $\mathbf a^1,\dots, \mathbf a^n\in A^k$.

The notion of \emph{polymorphism} plays the central role in the 
algebraic approach to the $\CSP$. 

\begin{definition}
  Given a $\Gamma$-structure $\mathbf{A}$, an $n$-ary
  \emph{polymorphism} of $\mathbf{A}$ is an $n$-ary operation $f$ on
  $A$ such that $f$ preserves the relations of $\mathbf A$. That is,
  if $\mathbf{a}^1,\dots,\mathbf{a}^n\in R$, for some $k$-ary relation
  $R$ in $\Gamma$, then $f^{(k)}(\mathbf a^1,\dots,\mathbf
  a^n)\in R$.  
\end{definition}

Given any relational structure $\mathbf{A}$, we can associate to it the algebra $\mathbb{A}=(A; Pol(\mathbf{A}))$. For that reason, instances of $\operatorname{CSP}(\mathbf{A})$ can be viewed as a instances of the constraint satisfaction problem $\operatorname{CSP}(\mathbb{A})$, if we define the instances of  $\operatorname{CSP}(\mathbb{A})$, to be those $\mathcal{I}=(V,A,\mathcal{C})$ all of whose constraint relations are preserved by all operations in the signature of $\mathbb{A}$. In that case, we say that an instance of $\operatorname{CSP}(\mathbb{A})$ is \emph{parametrized} by the algebra $\mathbb{A}$.

\begin{definition} A ternary operation $m:A^3\rightarrow A$ on a finite set is said to be \emph{Maltsev} if it satisfies the algebraic identities $m(x,x,y)\approx m(y,x,x)\approx y.$
\end{definition}

Clearly, any Maltsev operation $m(x,y,z):A^3\rightarrow A$, since it satisfies $m(x,x,x)\approx x$, for all $x\in A$. Therefore, every singleton $\{x\}$ is a subalgebra of $\mathbb{A}=(A;\{m(x,y,z)\})$.

\begin{example}
A typical example of a constraint satisfaction problem over a finite Maltsev template is the problem of solving a system of linear equations in $n$ variables over a fixed finite field $K$. Let $S$ be its solution space viewed as an $n$-ary relation on $K$. The operation $m(x,y,z)=x-y+z$ is a polymorphism of the relational structure $\mathbf{S}=(K; S)$. The converse is also true: namely, one can show that any $n$-ary relation on $K$, for $n\geq 1$, which has $m(x,y,z)$ as its polymorphism is a solution of some system of linear equations over $K$ in $n$ variables.
\end{example}

Finally, we state some facts about subdirect products of Maltsev algebras which will be needed later.

The first fact concerns the connectivity in subdirect products of simple Maltsev algebras. For the proof, see e.g. \cite{Burris1981}

\begin{theorem} \label{maltsev} Let $\mathbb{A}_1,\ldots,\mathbb{A}_n$ be simple algebras in a Maltsev variety. If
\[\mathbb{B}\leq_{sp} \mathbb{A}_1\times\ldots\times\mathbb{A}_n\]
is a subdirect product, then
\[\mathbb{B}\cong \mathbb{A}_{i_1}\times\ldots\times\mathbb{A}_{i_k}\]
for some $\{i_1,\ldots,i_k\}\subseteq \{1,\ldots,n\}$. 

In particular, if $\mathbb{A}$ and $\mathbb{B}$ are two simple Maltsev algebras then any subdirect product $\mathbb{C}\leq_{sp} \mathbb{A}\times\mathbb{B}$ is either the direct product or the graph of an isomorphism $f:\mathbb{A}\rightarrow\mathbb{B}$.
\end{theorem}

The existence of a Maltsev operation implies the following property on any subdirect product, which we will refer to as the \emph{rectangularity property}:

\begin{proposition}\label{rect} Let $C\leq_{sp} \mathbb{A}\times\mathbb{B}$, where $\mathbb{A}$ and $\mathbb{B}$ are algebras in a Maltsev variety. Then, the following holds: if $(a,b), (a,b'), (a',b')\in C$, then $(a',b)\in C$.
\end{proposition}

\begin{proof} Let $m(x,y,z)$ be a Maltsev operation on both $\mathbb{A}$ and $\mathbb{B}$. Then,
$$(a',b)=(m(a,a,a'), m(b,b',b'))\in C.$$
\end{proof}

Using Proposition \ref{rect}, one can prove a generalization of the second part of  Theorem \ref{maltsev}:

\begin{theorem} \label{link}  If $\mathbb{A}$ and $\mathbb{B}$ are two Maltsev algebras, with $\mathbb{B}$ being simple, then any subdirect product $\mathbb{C}\leq_{sp} \mathbb{A}\times\mathbb{B}$ is either the direct product or, there exists a maximal congruence $\theta\in\operatorname{Con}(\mathbb{A})$,  such that $C$ is the graph of an isomorphism $f:\mathbb{A}/\theta\rightarrow\mathbb{B}$.
In particular, the congruence $\theta$ is defined as follows: for $a,a'\in A$, $(a,a')\in \theta$ if, and only if, there exists $ b\in B $ such that $(a,b),(a',b)\in C.$
\end{theorem}

In the first case in the statement of Theorem \ref{link}, we refer to the subdirect product as \emph{linked} while, in the second case, we will say that $C$ is \emph{unlinked}.

The algorithm presented in the article will be based on the structure of minimal (i.e. strictly simple) algebras in a Maltsev variety.  We may assume that such a variety is idempotent, with the only relevant operation being the idempotent ternary Maltsev term $m(x,y,z).$ Using the classification of strictly simple algebras implied by the results from \cite{hobbymckenzie}, the minimal algebras appearing in an idempotent Maltsev variety can be one of the two following types:

\begin{enumerate}
\item an algebra polynomially equivalent to a two-element Boolean algebra; or
\item a simple affine algebra, which is polynomially equivalent to a module over a finite ring. In particular, such algebras can be viewed as simple affine groups of cardinality $|p^k|$ ($p$-a prime, $k\geq 1$), under the affine group operation $m(x,y,z)=x-y+z$, generated by any two of its distinct elements.
\end{enumerate}

More specifically,  every minimal algebra in an idempotent Maltsev variety can be viewed as an Abelian $p$-group, once an element of the algebra is designated as the zero element 0, and the affine operation gives rise to the usual Abelian group operations: $x+y=m(x,0,y)$ and $-x=m(0,x,0)$. In the case of algebras polynomially equivalent to the 2-element Boolean algebra, there is also an underlying polynomial operation, which is equivalent to $\oplus$, the addition modulo 2.

\section{\textsc{MOD}-Logspace Classes}\label{complexity}

In this article, we will be primarily concerned with problems which can be placed in the computational complexity class  \textsc{L} (deterministic logspace) and the functions which can be computed using Turing machines operating in deterministic logspace and using oracles of counting nature, which are intimately related to the computational complexity of common problems in linear algebra over a finite ring.

A typical example of a problem which can be solved in deterministic logspace is the problem $(s,t)$-\textsc{UCONN}, the connectivity in undirected graphs, which, given two vertices $s$ and $t$ in the input graph, asks whether there exists a path between $s$ and $t$. This rather deep  fact will be used in the analysis of the complexity of the algorithm presented here and is due to Reingold (\cite{reingold2008undirected}.)

A \emph{logspace transducer} is a Turing machine with a read-only input tape, a write-only
output tape, and a worktape which can contain at most $\mathcal{O}(\log n)$ symbols at any time. For that reason, one can view a logspace transducer as a function $F$ mapping instances of an algorithmic problem $\mathfrak{P}_1$ into instances of an algorithmic problem $\mathfrak{P}_2$ so that, if $\mathcal{P}$ is an instance of the problem $\mathfrak{P}_1$, the Turing machine for the function $F$, computing the instance $\mathcal{J}=F(\mathcal{P})$ of $\mathfrak{P}_2$,  operates in deterministic logspace.

It is a fairly elementary fact from theory of computational complexity that a composition of two logspace transducers is again a logspace transducer. For a proof, see e.g. \cite{arora2009computational}.

In our algorithm, logspace transducers play an important role. The algorithm will be based on a series of reductions from an instance to another one and those reductions will be carried out by logspace transducers which have access to a particular type of oracle, which will place the complexity of our algorithm in a particular complexity class, closely connected to the complexity of problems in linear algebra over finite rings $\mathbb{Z}_k$, $(k\in\mathbb{Z})$.

To explain this connection, we need to introduce the notion of particular subclasses of logspace with counting, the classes $\textsc{MOD}_k\textsc{L}.$ 

The complexity class \#\textsc{L} consists of all computable functions $f:\mathbb{N}\rightarrow\mathbb{N}$, such that there is a nondeterministic logspace-bounded Turing machine $T$, which halts on every input and along every computation path, so that the number of accepting paths on input $x$ is $f(x)$. For an integer $k\geq 2$, the class $\textsc{MOD}_k\textsc{ L}$ is defined to be the class of sets $A$ for which a function $f(x)\in \#$ \textsc{L} exists, such that, $x\in A$ if, and only if $f(x)\not \equiv 0 \mod k$.

On the other hand, the complexity class \textsc{DET}  consists of all problems that are $\textsc{NC}^1$-reducible to the problem of computing a determinant with entries from $\mathbb{Z}$, the ring of integers. We will not present a full definition of $\textsc{NC}^1$-reductions here but it would suffice to say that every reduction, which can be carried out by a  logspace transducer is an $\textsc{NC}^1$-reduction. This is a direct consequence of the result from the complexity theory that $\textsc{L}\subseteq\textsc{NC}^1$ (see e.g. \cite{arora2009computational}.) We can consider proper subclasses of \textsc{DET}, $\textsc{DET}_k$, for every $k\geq 2$, which consist of all the problems, which can be $\textsc{NC}^1$-reduced to the problem of computing the determinant of a matrix over $\mathbb{Z}_k$.

Typical examples of problems which are complete for the class \textsc{DET} are the standard problems in linear algebra over $\mathbb{Z}$: rank computations, computing the determinant of a matrix, computing a solution of a linear system, computing a basis of a kernel of a linear transformation, computing the inverse of a matrix, etc. For the subclasses $\textsc{MOD}_k\textsc{ L}$, all these problems are still complete problems when relativized to matrices over $\mathbb{Z}_k$.

For more information on \textsc{MOD}-logspace classes and their properties, the reader is referred to \cite{buntrock1992structure} and \cite{hertrampf2000note}.

We will need the following fact:

\begin{theorem} (\cite{buntrock1992structure}) For $k\geq 2$,

$$\textsc{L}^{\textsc{MOD}_k\textsc{ L}}\subseteq \textsc{DET}_k.$$
\end{theorem}

That is, every algorithmic problem which can be solved in logspace with the use of oracles which are in $\textsc{MOD}_k\textsc{ L}$ can be reduced to the problem of computing the determinant of a matrix over $\mathbb{Z}_k$,.

Our main result will be the following theorem:

\begin{theorem}  Let $\mathbb{A}$ be a finite Maltsev algebra. Then, $\operatorname{CSP}(\mathbb{A})\in \textsc{DET}$.
\end{theorem}

The proof will be based on the construction of an algorithm for $\operatorname{CSP}(\mathbb{A})$ which consists of the composition of a polynomial number of logspace transducers, which utilize oracles in $\textsc{MOD}_{k_1}\textsc{ L}, \ldots, \textsc{MOD}_{k_m}\textsc{ L}$, where $k_1,\ldots,k_m$ are integers which are completely determined by $\mathbb{A}$.

\section{Symmetric Datalog and (1,2)-Consistency}\label{datalog}

A \emph{Datalog program} for a relational template $\mathbf{A}$ is a finite set of rules of the form $T_0\leftarrow T_1,T_2,\ldots, T_n$ where $T_i$'s are atomic formulas. 
$T_0$  is the \emph{head} of the rule, while $T_1,T_2,\ldots, T_n$ form the \emph{body} of the rule.
Each Datalog program consists of two kinds of relational predicates:
the \emph{intentional} ones (IDBs), which are those occurring at least once in the head of some rule, and which are not part of the original signature of the template (they are derived by the computation).
The remaining predicates are said to be the \emph{extensional} ones, or EDBs. They are relations from the signature of the template and do not change during computation; i.e. they cannot appear in the head of any rule.
In addition to those, there is one special, designated IDB, which is nullary (Boolean) and referred to as the \emph{goal} of the program. The semantics of Datalog programs are generally defined
in terms of fixed-point operators. 

 A rule $T_0\leftarrow T_1,T_2,\ldots, T_n$ is said to be \emph{linear} if at most one atomic formula in its body is an IDB. A Datalog program is said to be \emph{linear} if all its rules are linear. The evaluation of a linear Datalog program is in nondeterministic logspace since, from the ccomputational complexity point of view, it reduces to repeated connectivity checks in a finite directed graph corresponding to the program.The linear rules in which an IDB appears in the body are said to be \emph{recursive}.

The \emph{symmetric complement} of a recursive linear rule  $T_0\leftarrow T_1,T_2,\ldots, T_n$ in which, without loss of generality,  the IDBs are $T_0$ and $T_1$, is defined to be the rule
$$T_1\leftarrow T_0,T_2,\ldots, T_n.$$ If the rule is non-recursive, its symmetric complement is the rule itself.

A linear Datalog program is \emph{symmetric}, if the symmetric complement of every rule also appears in the program.

Given a CSP instance $\mathcal{I}=(V, A, \mathcal{C})$ over a relational template $\mathbf{A}$, its \emph{canonical Symmetric}  (1,2)-\emph{Datalog program} has a unary IDB, $P_i(x)$, for each domain $P_i\subseteq A$ in the instance, with  the said IDBs being the only IDBs in the program. The program allows derivation rules with the body involving at most two variables along with their symmetric complements. 1 indicates that the program is deriving facts about unary relations only, while $2$ indicates that the maximal number of distinct variables in any rule is 2.  

If the instance $\mathcal{I}$ is such that:

\begin{enumerate}
\item all constraints are binary; and
\item for every pair of variables $i,j\in V$, there exists a unique binary constraint $E_{i,j}\in \mathcal{C},$ so that $E_{i,j}=E_{j,i}^{-1}$
\end{enumerate}

if the canonical symmetric (1,2)-Datalog program does not derive a contradiction (empty instance), all constraints  in the derived instance will be subdirect products. This is due to the fact that the only recursive rules involving unary IDBs $R_i, R'_i$ and $R_j, R'_j$, defined  on $P_i$ and $P_j$, respectively, are of the form 
$R_j (x)\leftarrow E_{i,j}(y,x),R_i(y)$, $R'_j(x)\leftarrow E_{j,i}(x,y),R'_i(y)$, along with their symmetric complements.

Finally, the complexity of computing the canonical symetric (1,2)-Datalog program for any (not necessarily binary) instance $\mathcal{I}$ of $\operatorname{CSP}(\mathbf{A})$ is in the deterministic logspace (by \cite{reingold2008undirected}) since the run of the program can be viewed as a sequence of connectivity checks in an  undirected graph, associated with $\mathcal{I}$. This fact will play an important role in the construction of a \textsc{DET} algorithm for solving instances of CSPs parametrized by finite Maltsev algebras.

\section{The Main Algorithm: A Proof of Theorem 3.2}\label{outline}

The construction of the algorithm is substantially motivated by a similar algorithm in \cite{delic2025constraint}. In that article, an algorithm was presented, which placed the complexity of the general CSP over finite templates in the class \textsc{MOD}, in the special case when the multisorted binary structure (template), whose construction will be presented in Subsection \ref{Binary} is a homomorphism core.  After that, an exhaustive analysis of 2-generated subuniverses which appear inside the sorts (domains) of the binary structure, reduces the problem to the case when all 2-generated subuniverses that potentially contain a solution have a simple affine quotient in addition to possible trivial subuniverses which may contain solutions to the instance. After that reduction,  the notion of \#\textsc{L}-consistency is developed and it is shown that a sequence of such consistency checks leads either to an empty instance or one in which every element contains a solution.  Since all consistency checks were in the deterministic logspace with counting, the entire algorithm had the same complexity. In fact, the algorithm was in the class \textsc{DET}, which, roughly speaking, encompasses those problems that can be solved by computing determinants over finite rings $\mathbb{Z}_k$, for some $k\geq 2$.

Here, we will present a simplified version of this algorithm, which will utilize $\textsc{L}$-consistency,  but the sequence of consistency checks will be limited to minimal subuniverses, and not all 2-generated subuniverses. A significant improvement on the original algorithm is the fact that we can completely dispense with the requirement that the binary multisorted template be a core.  This is due to the fact that subdirect products of two simple Maltsev algebras have a very predictable, symmetric structure.

Let $\mathcal{I}$ be an instance of $\operatorname{CSP}(\mathbb{A})$, where $\mathbb{A}$ is a finite Maltsev algebra. 

The first step in the algorithm reduces the instance $\mathcal{I}$ to a syntactically simple binary instance $\mathcal{P}$, parametrized by the Maltsev algebra $\mathbb{A}'=\mathbb{A}^{\lceil\frac{p}{2}\rceil}$, where $p$ is the maximal arity of a relation in $\Gamma$, the relational signature of the template $\mathbf{A}$. 

\subsection{Reduction to the Binary Case}\label{Binary}

In this subsection, we outline a reduction of an instance $\mathcal{I}$ of $\operatorname{CSP}(\mathbf{A})$, where $\mathbb{A}$ is a finite idempotent algebra, to a binary instance over a binary relational template, parametrized by $\mathbb{A}^m$, for some $m\geq 1$. The construction is due to L. Barto and M. Kozik and we largely adhere to their exposition in \cite{b-k1}. The reduction is  given by first-order (in fact, quantifier-free) formulas in a bounded number of variables, and can be carried out in deterministic logarithmic space.

An instance $\mathcal{P}$ is said to be \emph{syntactically simple} if it satisfies the following conditions:

\begin{itemize}
\item Every constraint in $\mathcal{C}$  is binary and its scope is a pair of  variables $(x,y)\in V^2$.
\item For every pair of variables $x,y$, there is precisely one constraint $E_{x,y}$ with the scope $(x,y)$.
\item If $x=y$, then $E_{x,x}=\{(a,a)\,\vert\, a\in P_x\}$, where $P_x$ is the $x$-th domain.
\item If $(x,y)$ is the scope of $E_{x,y}$, then $(y,x)$ is the scope of the constraint $E_{y,x}=\{(b,a) \, \vert\, (a,b)\in E_{x,y}\}$ (\emph{symmetry of constraints}).
\end{itemize}

Given any finite algebra $\mathbb{A}$ parametrizing the instance $\mathcal{I}$ such that the maximal arity of a relation in $\mathcal{I}$ is $p$, we define a new, syntactically simple instance $\mathcal{P}$ in the following way:

\begin{itemize} 
\item The instance is parametrized by $\mathbb{A}^{\lceil \frac{p}{2}\rceil}$, which is an algebra satisfying all term identities (equations) $s\approx t$, satisfied by $\mathbb{A}$. 
\item For every $\lceil\frac{p}{2}\rceil$-tuple of variables in $\mathcal{I}$, we introduce a new variable in $\mathcal{P}$ and, if $x=(x_1,\ldots,x_{\lceil\frac{p}{2}\rceil})$ and $y=(y_1,\ldots,y_{\lceil\frac{p}{2}\rceil})$ with $x\neq y$, we introduce a constraint 
\begin{align*}
E_{x,y} &=\{((a_1,\ldots,a_{\lceil\frac{p}{2}\rceil}),(b_1,\ldots,b_{\lceil\frac{p}{2}\rceil}))\,\vert \, \\  & (a_1,\ldots,a_{\lceil\frac{p}{2}\rceil},b_1,\ldots,b_{\lceil\frac{p}{2}\rceil})\\
&\mbox{ is a $p$-assignment of values which satisfies all  }\\ & \mbox{ atomic formulas 
on the tuples of variables } x,y \}
\end{align*}
while, if $x=y$, the relation $E_{x,x}$ is simply the equality of $\lceil\frac{p}{2}\rceil$-tuples in $\mathbb{A}^{\lceil\frac{p}{2}\rceil}$.
\end{itemize}

The binary instance $\mathcal{P}$ constructed in this way will have a solution if, and only if, the instance $\mathcal{I}$ has a solution. 

From the reduction described above, it is easily seen that, if $\mathcal{I}$ is an instance parametrized by a finite algebra $\mathbb{A}$, then the constructed, syntactically simple binary instance $\mathcal{P}$ can be parametrized by the direct product $\mathbb{A}^{\lceil\frac{p}{2}\rceil}$. In particular, if the original instance $\mathcal{I}$ is parametrized by a Maltsev algebra, then so is $\mathcal{P}$.

Also, based on the discussion in Section \ref{datalog}, if $\mathcal{P}$ is a syntactically simple binary instance of $\operatorname{CSP}(\mathbb{A})$, the instance $\mathcal{P}^\ast$ of $\operatorname{CSP}(\mathbb{A}')$ computed by the canonical symmetric (1,2)-Datalog program for $\mathcal{P}$ will have the property that all its constraints will be subdirect products.

From this point on, we will adhere to the following notation:

\begin{itemize}
\item the set of variables of $\mathcal{P}$ will be denoted $V$ and it will be an initial segment of the set of positive integers: $V=\{1,\ldots,n\}=[n]$, for some $n\geq 1$. We also assume that $V$ is an ordered set, via the natural ordering $<$  of positive integers.
\item the domains of $\mathcal{P}$, which are all subuniverses of $\mathbb{A}'$, constructed during the reduction, will be denoted $P_i$ ($1\leq i\leq n$), with the convention that $P_i$ is the domain corresponding to the variable $i\in V$.
\end{itemize}

We may also assume that, as a part of the input, the algorithm has access to the catalogue of all minimal,2-generated subalgebras $\operatorname{Sg}(a,b)$, of $\mathbb{A}'$, where $a\neq b, a,b\in A'$. We will also use the notation $\operatorname{Sg}_i(X)$ to denote the subuniverse of $\mathbb{P}_i$, generated by a subset $X$ of $P_i= \mathbb{A}'$.

%The proof of the existence of the logspace algorithm using finitely many oracles of the type $\textsc{MOD}_k\textsc{L}$ (and, therefore, in \textsc{DET}) for solving syntactically simple binary instances $\mathcal{P}$ of $\operatorname{CSP}(\mathbb{A}')$ will be inductive on the $s=\max \{|P_i| \, : \, i\in V\}$. The case $s=1$, i.e. when all domains are singletons or empty sets is trivial.

%We assume that there exists a logspace algorithm $\mathcal{A}_{s-1}$ with access to finitely many oracles of the type $\textsc{MOD}_k\textsc{L}$,  which decides all syntactically simple binary instances of $\operatorname{CSP}(\mathbb{A}')$ in which all domains have size at most $s-1$, for some $s\geq 2$ and show how to construct an algorithm $\mathcal{A}_s$ deciding instances whose domains are of size at most $s$ and which is in \textsc{DET}.

\subsection{Triple Graph of a Syntactically Simple Binary Instance}\label{graph}

Let $\mathcal{P}$ be a syntactically simple binary instance of $\operatorname{CSP}(\mathbb{A}')$, where $\mathbb{A}'$ is a finite Maltsev algebra. To every such instance, we can associate an undirected graph $\mathbf{G}_{\mathcal{P}}$, defined in the following way:

\begin{itemize}
\item the vertices of $\mathbf{G}_{\mathcal{P}}$ are the triples $(i,a,b)$, such that $i\in V$ and $a,b\in P_i$, $a\neq b$.
\item let $a,b\in P_i$, ($a\neq b$) and $c,d\in P_j$ ($c\neq d$), with $i\neq j$. Then, $(i,a,b)$ and $(j,c,d)$ will form an edge  in $\mathbf{G}_{\mathcal{P}}$ if, and only if the following is true:
$$(a,c),(b,d)\in E_{i,j}, \quad \mbox{ and } \quad (a,d),(b,c)\not\in E_{i,j}.$$
[or, equivalently, $(c,a),(d,b)\in E_{j,i}$, and  $(d,a),(c,b)\not\in E_{j,i}.$]
\end{itemize}

From the definition of the edge relation in $\mathbf{G}_{\mathcal{P}}$,  it is easily seen that $(i,a,b)$ and $(j,c,d)$ form an edge in $\mathbf{G}_{\mathcal{P}}$ if, and only if, $(i,b,a)$ and $(j,d,c)$ form an edge. For that reason, we can identify the vertices $(i,a,b)$ and $(i,b,a)$, for any $i\in V$ and $a,b\in P_i$, with $a\neq b$. In the remainder of the paper, we will tacitly assume that such an identification has been made.

The graph $\mathbf{G}_{\mathcal{P}}$ can be constructed in logspace, given the instance $\mathcal{P}$. If $(i,a,b)$ is a vertex of $\mathbf{G}_{\mathcal{P}}$, its connected component in this triple graph will be denoted $C(i,a,b)$.

We point out that every edge of the graph $\mathbf{G}_{\mathcal{P}}$ represents a subdirect product of two 2-generated subuniverses of the Maltsev algebra $\mathbb{A}'$. In particular, every edge in the triple graph represents an unlinked subdirect product of two such algebras. In the case when $\operatorname{Sg}_i(a,b)$ and $\operatorname{Sg}_j(c,d)$ are both simple algebras, the edge between $(i,a,b)$ and $(j,c,d)$ implies the existence of an isomorphism between the two algebras mapping $a$ to $c$ and $b$ to $d$.

The motivation for the definition of the triple graph associated to a syntactically simple binary instance  was provided by \cite{egri2014space}, where a similar definition of the notion of a triple digraph associated with a conservative digraph was used to prove that CSPs over conservative digraph templates which admit Hagemann-Mitschke terms can be solved in deterministic logspace.

\subsection{Construction of the Algorithm}

In this subsection, we will present the algorithm for solving Maltsev constraint satisfaction problems over finite Maltsev templates, whose construction and a proof of correctness will take up the remainder of the article. We have already shown that, given an instance $\mathcal{P}$ of $\operatorname{CSP}(\mathbb{A})$, where $\mathbb{A}$ is a finite Maltsev algebra, the problem of determining whether $\mathcal{P}$ contains a solution can be logspace reduced to the problem of determining whether a syntactically simple binary instance $\mathcal{P}$ of $\operatorname{CSP}(\mathbb{A}')$, where $\mathbb{A}'=\mathbb{A}^{\lceil \frac{p}{2}\rceil}$ and $p$ is the maximal arity of a constraint relation in $\mathcal{P}$. Furthermore, we may assume that the first step of the algorithm conists of the application of the canonical symmetric (1,2)-Datalog program to $\mathcal{P}$ and that the instance is (1,2)-consistent, assuming the output of the program is a nonempty instance $\mathcal{P}'$. If the output is the empty instance, the algorithm terminates with the output `\textsc{No solution}'. Thus, we may assume that, if the instance $\mathcal{P}'$ is nonempty, $\mathcal{P}=\mathcal{P}'$. The canonical symmetric (1,2)-program operates in deterministic logspace, therefore, the complexity of the algorithm at this stage is in deterministic logpsace.

As the next step, we construct the triple graph $\mathbf{G}_{\mathcal{P}}$, associated with the instance $\mathcal{P}$, as described in Subsection \ref{graph}. The construction of $\mathbf{G}_{\mathcal{P}}$ can be carried out in logspace.

%The algorithm augments the language by introducing, for every $i\in V$, a new binary predicate $S_i(x,y)$, which are initialized to be $S_i=P_i\times P_i$.

Given an original instance of a CSP with a Maltsev template, as we have shown earlier, it can be reduced to a syntactically simple instance (multiconsistency graph) $\mathcal{P}$ in deterministic logspace. In fact, this reduction is a first-order, quantifier-free one. Next, we show how to interpret the multiconsistency graph, along with the triple graph into a single structure, from which both structures can be recovered, and which will be an input into our algorithm, whose complexity will be shown to be in $\#\textsc{L}$. 

The domain of the new instance will consist of triples of the form $(i,a,b)$, where $1\leq i\leq n$ and $a,b\in A$, where we allow for the possibility that $a=b$. The triples of the form $(i,a,a)$, where $a\in A$, will correspond to the elements of the original syntactically simple instance $\mathcal{P}$. Within the set of the triples of the new structure, they can be defined as those triples $(i,x,y)$, which satisfy the quantifier-free formula $x=y$. Similarly, one can define binary constraint relations $E_{i,j}$ among them, using the corresponding relations from $\mathcal{P}$. Finally, the edges of the triple graph can be defined from the original binary constraints $E_{i,j}$ ($1\leq i,j\leq n$), again using quantifier-free first-order formulas. Clearly, this reduction can be accomplished in deterministic logspace, since it is given by a quantifier-free first-order interpretation. From now on, we will refer to the structure encoding both the syntactically simple instance and its accompanying triple graph as $\mathcal{P}$.

We may assume that we are also given a listing of all elements of $\mathcal{P}$, which either represent vertices of the syntactically simple instance or a minimal subuniverse of $\mathbb{A}$: $L=\{(i,a,b)\, : \, i\in[n], a,b \in P_i\}$., where either $a=b$ or $Sg_{P_i}(a,b)$ is a minimal subuniverse.

The key ingredient of our algorithm is a particular consistency check, the $(i,a,b)$-\emph{test}, whose execution complexity is in the deterministic logspace with counting, or, more specifically in the class $\textsc{MOD}_k\textsc{L}$.  The algorithm examines all triples $(i,a,b)\in L$, and,  if the $(i,a,b)$-test fails for the triple on which the test has been applied, the triple is removed from the list $L$. The sequence of tests continues until either (1) the list $L$ becomes empty and the instance has no solutions; or (2) the list undergoes no further changes. In the proof of the algorithm correctness, we will show that, in case (2), a solution exists through every element of the pairs remaining in the list $L$.

\subsection{$(i,a,b)$-Test}

Let $(i,a,b)\in L$ be a triple., with $a\neq b$. The $(i,a,b)$-test stars by constructing the \emph{group} $J_{i,a,b}$ of $(i,a,b)$, in the following way: let $j\in V$  be such that, for some $(j,c,d)\in P_j$, $(c\neq d)$, there exists a path in the triple graph from $(i,a,b)$ to $(j,c,d)$. Then, $j\in J_{i,a,b}$. From the fact that the existence of a path in an undirected graph can be verified in deterministic logspace, we get the following proposition.

\begin{proposition} For every $i\in V$, $a, b\in P_i$ $(a\neq b)$, the group $J_{i,a,b}$ of $(i,a,b)$ can be computed in deterministic logspace. 
\end{proposition}

After the group $J_{i,a,b}$ of $(i,a,b)$ has been computed, along with the restriction of the triple graph to the domains in $J_{i,a,b}$ (we tacitly assume that $i\in J_{i,a,b}$):

\begin{enumerate}
\item Run the (1,2)-symmetric Datalog program on the instance $P'$, in which, $P'_i$ is $P_i$ restricted to $\operatorname{Sg}_{P_i}(a,b)$, where we allow for the possibility that $a=b$. If the program results in an empty instance, $(i,a,b)$-test fails and $(i,a,b)$, $(i,a,a)$, and $(i,b,b)$ are removed from $L$ and both $a$ and $b$ are deleted from $P'_i$; or
\item If the (1,2)-program succeeds and terminates in a non-empty, nontrivial  (1,2)-consistent subinstance, whose $i$-th domain is $\operatorname{Sg}_{P_i}(a,b)$, where $a\neq b$, with $|Sg_{P_i}(a,b)|=p^k$, where $p$ is a prime and $k\geq 1$, consider the triple graph restricted to domains of the resulting instance, whose indices are in $J_{i,a,b}$.  For every $c\in P'_j$ ($j\in P_{i,a,b}$, $j\neq i$), count the number $l$ of elements of $P'_i$ connected to it in the triple graph. If $l\equiv 1 (\mod p)$, the triple $(i,a,b)$ has passed the test. Otherwise, if $l\equiv 0 (\mod p)$, we delete $(i,a,b)$ from the list $L$; otherwise
\item If the (1,2)-symmetric Datalog succeeds, but $\operatorname{Sg}_{P_i}(a,b)$ has been reduced to a trivial subalgebra $\{c\}$, 
$c\in\operatorname{Sg}_{P_i}(a,b)$, the triple $(i,a,b)$ as well as all triples $(i,d,d)$, where $d\neq c$ and $\operatorname{Sg}_{P_i}(a,b)$ are deleted from $L$. Also, all such $d$ are deleted from $P_i'$.

\end{enumerate}

If the (1,2)-symmetric Datalog fails, there cannot exist a solution through any element of $\operatorname{Sg}_{P_i}(a,b)$, since (1,2)-consistency fails. Thus, we assume the (1,2)-symmetric Datalog test succeeds.  If the test reduces the $i$-th domain to a proper subuniverse of $\operatorname{Sg}_{P_i}(a,b)$, such a subuniverse must be trivial and will be dealt with through an $(i,c,c)$-test at some, for some $c$. In this case, we also remove $(i,a,b)$ from $L$.

Consider $j\in J_{i,a,b}$, where $J_{i,a,b}$ is the group of $(i,a,b)$. Then, there exists a path,  call it $\alpha$, in the triple graph whose endpoints are $(i,a,b)$ and $(j,c,d)$, where $c\neq d$.  For every $j\in J_{i,a,b}$, we can fix such a path $\alpha_j$ that witnesses the membership of $j$ in $J_{i,a,b}$.  For $j\in J_{i,a,b}$ and its witness path $\alpha_j$, there exists a maximal congruence $\theta_j$ on $P'_j$, the $j$-th domain of the resulting (1,2)-consistent instance, such that
$$P'_j/\theta_j \cong \operatorname{Sg}_{P_i} (a,b).$$

Namely, the congruence $\theta_j\leq P'_j\times P'_j$ can be defined in the following way: $(c,d)\in \theta_j$  if, and only if  $\exists e\in \operatorname{Sg}_{P_i}(a,b)$ such that $e$ is connected via $\alpha$-paths to $c$ and $d$.

Using the rectangularity property of Maltsev subdirect products, it can be easily verified that $\theta_j$ is indeed a congruence on $P'_j$, and that the corresponding quotient algebra is isomorphic to $P'_i=\operatorname{Sg}_{P_i}(a,b)$.

Therefore, for every $j\in J_{i,a,b}$, there is a congruence $\theta_j$, such that $P'_j/\theta_j$ is isomorphic to $P'_i=\operatorname{Sg}_{P_i}(a,b)$. For that reason we can view the restriction of the (1,2)-consistent instance P' to the domains from $J_{i,a,b}$ as an instance of a constraint satisfaction problem over a simple affine group $Sg_{P_i}(a,b)$. As mentioned earlier, by fixing a particular element (coset) in every $P'_j/\theta_j$ as the zero element, we can view this instance as a system of linear equations over a fixed Abelian $p$-group $A$. In this system, every equation involves two variables and can be viewed as an equation of the form $x_j+c=x_k$, for some $j,k\in J_{i,a,b}$ $(j\neq k)$ and $a\in A$.

From (1,2)-consistency of the derived instance,  and every $j_1\in J_{i,a,b}$ and $c\in P'_j$, and every path pattern $\alpha$ of the form $j_1j_2\ldots j_k=i$ ($j_1,j_2,\ldots,j_{k-1}\in J_{i,a,b}$) in the triple graph,  where the subdirect products $P'_{j_l}\times P'_{j_m}$ are all unlinked, that is graphs of isomorphisms between two copies of the simple affine group $A$, there must exist some $e\in P'_i=\operatorname{Sg}_{P_i}(a,b)$, such that $c$ and $e$ are endpoints of a path with this pattern. In fact, such path pattern $\alpha$ defines an isomorphism:

$$\phi_\alpha: P_j'/\theta_j \cong \operatorname{Sg}_{P_i}(a,b)=P'_i.$$

Now, suppose $c\in P'_j$ was in a block $C$ of $\theta_j$, corresponding to the triple graph path $\alpha_j$,  for which every element was connected to a unique element $f\in \operatorname{Sg}_{P_i}(a,b)$.  What happens if there is an $e\in \operatorname{Sg}_{P_i}(a,b)$ and a path $\beta$, in the triple graph, consisting of isomorphisms between copies of $A$ from $j$ back to $i$ connecting $c$ to $e$, where $e\neq f$? If this is indeed the case,  we have a closed path pattern consisting of $\alpha_j$ followed by $\beta$, which starts and ends in the domain $P'_i=\operatorname{Sg}_{P_i}(a,b)$, and which corresponds to a sequence of linear equations on $A$, which, using Gaussian elimination, results in an equation of the form $x_i+c=x_i$, where $c$ is a non-zero element of $A$. This is, clearly, a contradiction. So, in order for the (1,2)-consistent instance to have a solution,  every $c\in P'_j$ ($j\in J_{i,a,b}$, $j\neq i$), must be connected to precisely one element of $P_i=\operatorname{Sg}_{P_i}(a,b)$ in the triple graph. 

If such a $c\in P'_j$ is connected in the triple graph to at least two elements of $\operatorname{Sg}_{P_i}(a,b)$, we will show that it is connected to all elements of this algebra in the triple graph. So, assume that the path $\alpha_j$ connects some element $e$ of ${Sg}_{P_i}(a,b)$ to $c$ and that there is another path $\beta$ connecting $c$ to some $f\in {Sg}_{P_i}(a,b)$, where $e\neq f$. As stated in the preceding paragraph,  the two paths in the triple graph, the first one connecting e to c and then back to e, and the other one, connecting e to c and then to f. induce two automorphisms of $\operatorname{Sg}_{P_i}(a,b)$:

$$ \pi_1: x\mapsto x+g \mbox{ and } \pi_2: x\mapsto x+h, \quad \mbox{for some } g,h\in A.$$
Keeping in mind that $A$ is a minimal algebra, equivalent to the affine Abelian group, under the Maltsev operation $x-y+z$, it is easily seen that, since $g\neq h$, $g$ and $h$ generate the entire affine Abelian group $A$. This also implies that every automorphism $\pi'$ of $A$, which is of the form $x\mapsto x+d$, for some $d\in A$ is generated by $\pi_1$ and $\pi_2$ using the Maltsev affine operation.

In turn, every such automorphism corresponds to some closed path pattern from $i$ to $i$ in the triple graph. Therefore, all elements of $\operatorname{Sg}_{P_i}(a,b)$ will be reachable from $c$ in the triple graph. Thus, if $c\in P'_j$ is connected to at least two elements of $\operatorname{Sg}_{P_i}(a,b)$, it will be connected to \emph{all} elements of this minimal algebra in the triple graph. This entails that $l=p^k\equiv 0 (\mod p)$. 

  %Due to the minimality of the underlying algebra, the orbits of $\pi$, being substructures fixed by $\pi$, can either be singletons or there is a single orbit of cardinality $|A|=p^k$. The former cannot be the case since the permutation $\pi$ is not trivial, so there is only one orbit. But, in that case, by composing the path $\beta$ with powers of $\pi$, we see that there are paths in the triple graph, consisting of isomorphisms between copies of $A$, which connect $c$ to every element of ${Sg}_{P_i}(a,b)$ 

Therefore,  in the $(i,a,b)$-test, once (1,2)-consistency has been established and assuming that neither $a$ nor $b$ were eliminated in the process, counting the number $l$ of elements of $P'_i$ connected to an element $c\in P_j'$ in the resulting triple graph will determine if the solution exists. If $l\equiv 1 (\mod p)$, no inconsistency obstructing the existence of a solution through the minimal algebra is detected and the triple $(i,a,b)$ has passed the test.  Otherwise, $l\equiv 0 (\mod p)$ and no solution can exist through either $a$ or $b$.

In the case of a triple $(i,a,a)$, where $i\in V$ and $a\in P_i$, the $(i,a,a)$-test simply verifies the (1,2)-consistency of the instance whose $i$-th domain $P_i$ is $\{a\}$, using the (1,2)-symmetric linear Datalog program. If the test fails, $(i,a,a)$ is eliminated from the list $L$.

\subsection{ Proof of the Algorithm Correctness}

To complete our proof, we show that passing our $\#\textsc{L}$ consistency test is equivalent to the CSP having a solution. This subsection essentially consists of the proof of the following statement:

\begin{proposition}
    If, at the conclusion of all $(i,a,b)$-tests, the list $L$ is nonempty, the instance $P$ has a solution.
\end{proposition}

%\begin{proof}
    %If the list $L$ is nonempty, then, for every $i\in V$, there must exist some $(i,a,b)$ (with the possibility that $a=b$), which remains in $L$. Choose some $i_0\in V$ and $(i_0,a,b)$, which is in $L$ at the end of all consistency checks.  Since the $(i,a,b)$-test succeeded, there is a (1,2)-consistent subinstance $P'$, for which $P'_{i_0}=\operatorname{Sg}_{P_{i_0}}(a,b)$ is the $i_0$-th domain.

    %So suppose now that we have a domain $\D_{i_0}$ which is non-simple and we have our associated multiconsistency graph $G$ which passed the consistency checks. Define the group of a domain $\A_i$ to be the indices of the domains which can be reached from $\A_i$ along isomorphism paths. Let $\t_{i_0}$ be the maximal congruence associated with $\D_{i_0}$, let $I_0$ be the group of $\D_{i_0}/\t_{i_0}$, and let $\C$ be any block of $\D_{i_0}/\t_{i_0}$. We now have a $(1,2)$ consistent subinstance $G'$ where the ${i_0}^{th}$ domain is $\C$ and the other domains in the group are the isomorphic copy of $\C$ and the remaining domains are as before. Recall that since the subinstance passes the datalog check, every pair of domains forms a subdirect product. If this new instance passes our $\#\textsc{L}$ consistency check then the inductive hypothesis applies and we get a solution to $G'$. 
    %\end{proof}

   We proceed to prove the claim which will play the crucial role in establishing that no inconsistency i.e. unsolvability of a system of linear equations in two variables over an  Abelian $p$-group (for some prime $p$), which has been discovered through the construction of $L$,can be passed down to the subinstance $P'$.  
   
   Let $j\in V$ be such that $j\neq i_0$, and let $C=\operatorname{Sg}_{P'_j}(c,d)$, be the minimal subuniverse of $P'_j$, generated by two elements $c,d\in P'_j$, so that $(j,c,d)\in L$.
    
\begin{claim}\label{group}

    The group of $(j,c,d)$ in $P'$ is the same as the group $J_{j,c,d}$ of  $(j,c,d)$ computed during the $(j,c,d)$-test.  
    
\end{claim}

\begin{proof} Certainly $J_{j,c,d}$ must be contained in the group of $(j,c,d)$. in $P'$.  Next, assume there is an index $j'$ in the group of $(j,c,d)$ in $P'$, which is not in  $J_{j,c,d}$. This implies that between any domain whose index is in the group and the domain of index $j'$ we do not have any possible isomorphism. Thus by Theorem~\ref{maltsev} they are full direct products and so we have that it is still a full direct product when we reduce the $j$-th domain of $P'$,to $\operatorname{Sg}(c,d)$. Thus $j'$ cannot be in the group of $(j,c,d)$ in $P'$ and the claim is proven.
\end{proof}

Thus, solvability of any system of linear equations in two variables over a fixed elementary Abelian $p$-group, determined by a minimal agebra in one of the domains, is preserved by restricting the said domain to that minimal algebra, provided that the minimal algebra in question, $\operatorname{Sg}_{P'_j}(c,d)$, is such that $(j,c,d)$ is in $L$, after the termination of all $\#\textsc{L}$-consistency tests.

Using this claim, we can proceed through the list of variables $V=\{1,2,\ldots,n\}$, starting with $i=1$, restrict the domain to the minimal algebra $P'_1=\operatorname{Sg}_{P_1}(a_1,b_1)$, for some $a,b\in P_1$, such that $(1,a_1,b_1)\in L$. If the only pairs in $L$, such that $(1,a_1,b_1)\in L$ are those for which $a_1=b_1$, the restricted domain $P'_1$ will be the trivial algebra $\{a_1\}$. After that, we apply the (1,2)-symmetric linear Datalog program to the instance $P'$ with the restricted domain $P'_1= \operatorname{Sg}_{P_1}(a_1,b_1)$ and,  proceed to reduce the domains $P'_2,\ldots, P'_n$ of $P'$ in the analogous fashion, until we are left with an instance $\tilde{P}$, all of whose domains are either trivial or minimal algebras.

Due to the fact that $\tilde{P}$ is a (1,2)-consistent instance, all trivial domains can be ignored since, the binary constraints in which those variables appear, are represented by full direct products with other domains in the said constraints. Therefore, we can assume that all domains of $\tilde{P}$ are minimal algebras, i.e. that they are equivalent to elementary Abelian $p$-groups, for some primes $p$.

Consider some $i_0\in V$. Then, $\tilde{P}_{i_0}=\operatorname{Sg}_{P_{i_0}}(a,b)$, for some $a,b\in P_0$, where $a\neq b$.  The group of $\tilde{P}_{i_0}$ in the triple graph for $\tilde{P}$ will be the same as $J_{i_0,a,b}$ in the original instance $P$, due to Claim \ref{group}. In fact, the groups of variables of $\tilde{P}$ will form a partition of $V$. Every block of such a partition corresponds to a linear system of equations in two variables over a fixed Abelian group $A$.  Since the solvability of such systems cannot be altered when restricting domains using triples in $L$, every such system will have a solution, with solutions passing through every element of the minimal algebra, which is the domain.

\section{Conclusion and Open Problems}

In this article, we developed a new algorithm for solving instances of $\operatorname{CSP}(\mathbf{A})$, where $\mathbf{A}$ is a finite relational template parametrized by a Maltsev algebra. The interesting feature of the algorithm is that, in the presence of constraint relations of arity at most 2,  it is primarily based on the series of consistency checks (albeit, not local ones) for pairs of elements, and for that purpose, it utilizes an auxiliary structure, an undirected graph which can be associated to any instance with at most binary constraints. In fact, consistency testing relies heavily on the connectivity checks in the said graph, which underscores the key role played by logspace transducers in the algorithm construction. The ability to reduce the solvability of the binary instance to the new consistency checks also allows for a finer analysis of the computational complexity of such an algorithm.

The algorithm presented here sharply differs from the standard Bulatov-Dalmau algorithm (Generalized Gaussian Elimination) in the sense that it does not exploit the fact that the finite subpowers of finite Maltsev algebras have compact representations and avoids having to compute a representation of intersections of constraints. In fact, any algorithm for solving instances of Maltsev constraint satisfaction problems, which operates in logspace  cannot be based on maintaining generating sets, since the space requirements would be at least linear. Our algorithm, in addition to  logarithmic space on its work tapes, only needs a bounded number of \textsc{MOD}-logspace oracles, which can be fixed at the outset, since they depend solely on the parametrizing algebra $\mathbb{A}$, and which reduce to the problem of computing determinants over a fixed set of finite rings $\{\mathbb{Z}_{k_1},\ldots,\mathbb{Z}_{k_r}\}$, which, again, is solely determined by $\mathbb{A}$.

We believe that further analysis  of this algorithm may be helpful in answering the following question:

\begin{openproblem} If $\mathbf{A}$ is a finite relational template, which can be parametrized by a Maltsev algebra, can $\operatorname{CSP}(\mathbf{A})$ be defined in an extension of first-order logic such as Choiceless Polynomial Time with Counting (or, equivalently, Polynomial Interpretation Logic with H\"artig quantifiers), or in  \textsc{LFP+Rk}, the extension of the Least Fixed Point logic with the operators for computing ranks of matrices over finite fields $\mathbb{Z}_p$ ($p$ - a prime)?
\end{openproblem}

Namely, the algorithm presented here indicates that an instance of $\operatorname{CSP}(\mathbb{A})$, where $\mathbb{A}$ is a finite algebra which has a ternary Maltsev operation, can be solved by repeated use of reductions to a smaller instance which consist of two types of reductions: (1) reductions that can be implemented in Datalog (in fact, Symmetric Linear Datalog) and (2) those that are based on oracles which can be viewed as checking if certain matrices over rings $\mathbb{Z}_k$ are invertible. This, in our opinion, provides a strong evidence that answer to the problem above is likely to be affirmative, at least for the logic \textsc{LFP+Rk}. 

%Furthermore, it would be interesting to see if the algorithm developed here, which utilizes the triple digraph of the instance, could be adapted in such a way that it provides an alternative proof of the Dichotomy Theorem for syntactically simple binary instances of $\operatorname{CSP}(\mathbb{A})$, where $\mathbb{A}$ is an algebra with a Taylor (or Siggers) operation. In particular, if the triple digraph is constructed using 2-generated absorption-free subalgebras of $\mathbb{A}$, can our algorithm be adapted so that the connectivity conditions in the triple digraph can be used to determine the existence of a solution?

%\begin{openproblem} If $\mathbb{A}$ is a finite Taylor algebra, can an algorithm based on the triple graph of 2-generated absorption-free subuniverses of $\mathbb{A}$ be used to provide a proof of the Dichotomy Theorem?
%\end{openproblem}

\bibliography{digraph_reduction.bib}
\bibliographystyle{ws-ijac}

\end{document}